\documentclass{article}

\usepackage{lmodern}
\usepackage[utf8]{inputenc}
\usepackage[T1]{fontenc}
\usepackage[USenglish]{babel}
\usepackage{amsmath,amssymb,amsthm}

\usepackage{fullpage}

\usepackage{tikz} 
\tikzset{Vertex/.style={circle,thick,draw,fill=black,text=white,inner sep=0pt,minimum size=4pt}}

\usepackage{complexity}
\usepackage{enumitem}
\usepackage[labelsep=period,labelfont={sc},textfont=it,margin=\parindent,font=small,format=plain]{caption}
\usepackage[numbers,sort&compress]{natbib}
\usepackage[%
  pdftitle={Total 2-domination of proper interval graphs},%
  pdfauthor={Francisco J.\ Soulignac},%
  pdfcreator={},%
  pdfsubject={},%
  pdfkeywords={},%
  colorlinks=true,%
  linkcolor=blue,%
  citecolor=blue,%
  urlcolor=blue]{hyperref}  
\usepackage{doi}
\usepackage{cleveref}

\newcommand{\Fr}{f_r}
\newcommand{\FrImpl}{\hat{\Fr}}
\newcommand{\RR}{r}
\newcommand{\LL}{\ell}
\newcommand{\Ur}{u_r}
\newcommand{\Weight}{\omega}
\newcommand{\Pos}{\pi}
\newcommand{\Width}{\kappa}
\newcommand{\WidthImpl}{\hat{\Width}}
\newcommand{\Naturals}{\mathbb{N}}

\usepackage{accents}
\newcommand{\ubar}[1]{\underaccent{\bar}{#1}}

\newcommand{\ER}{\ubar{\bar{E}}}
\newcommand{\ES}{\ubar{E}}
\newcommand{\EB}{\bar{E}}
\newcommand{\BigE}[1]{\bar{#1}}
\newcommand{\SmE}[1]{\ubar{#1}}
\newcommand{\ReE}[1]{\BigE{\SmE{#1}}}

\newtheorem{theorem}{Theorem}

\crefname{equation}{}{}
\Crefname{equation}{}{}
\crefname{eqitem}{}{}
\Crefname{eqitem}{}{}
\crefname{step}{step}{steps}
\Crefname{step}{Step}{Steps}

\newlist{discription}{enumerate}{2}
\setlist[discription]{align=left,leftmargin=\parindent,labelsep=*,itemindent=!,labelindent=0pt,parsep=0pt,topsep=2pt,itemsep=2pt}

\begin{document}

\title{Total $2$-domination of proper interval graphs}

\author{Francisco J.\ Soulignac\thanks{CONICET and Departamento de Ciencia y Tecnología, Universidad Nacional de Quilmes, Bernal, Argentina.}~\thanks{Supported by PICT ANPCyT grant 2015-2419.}}

\date{\normalsize \texttt{francisco.soulignac@unq.edu.ar}}

\maketitle

\begin{abstract}
  A set of vertices $W$ of a graph $G$ is a total $k$-dominating set when every vertex of $G$ has at least $k$ neighbors in $W$.  In a recent article, Chiarelli et al.\ (Improved Algorithms for $k$-Domination and Total $k$-Domination in Proper Interval Graphs, Lecture Notes in Comput.\ Sci.\ 10856, 290--302, 2018) prove that a total $k$-dominating set can be computed in $O(n^{3k})$ time when $G$ is a proper interval graph with $n$ vertices and $m$ edges.  In this note we reduce the time complexity to $O(m)$ for $k=2$.  
  
  ~
  
  \noindent\textbf{Keywords:} total $2$-domination, straight oriented graphs, proper interval graphs. 
\end{abstract}

\section{Introduction}

A set of vertices $W$ of a graph $G$ is a \emph{total $k$-dominating} set when every vertex $v$ of $G$ has at least $k$ neighbors in $W$.  The problem of computing a total $k$-dominating set of $G$ with minimum cardinality is known to be \NP-complete for every $k \geq 1$, even when $G$ belongs to certain subclasses of chordal graphs~\cite{PradhanIPL2012} such as undirected path graphs~\cite{LanChangDAM2014,LaskarPfaffHedetniemiHedetniemiSJADM1984}.  In turn, when $G$ is an interval graph with $n$ vertices, the problem is solvable in $O(n^{6k+4})$ time, as recently proven by Kang et al.~\cite{KangKwonStrommeTelleTCS2017} (cf.\ \cite{ChiarelliHartingerLeoniPujatoMilanic2018}).  Moreover, the time complexity can be reduced to $O(n^{3k})$ when $G$ belongs to the subclass of proper interval graphs~\cite{ChiarelliHartingerLeoniPujatoMilanic2018}.  

Besides being a subclass of undirected path graphs, interval graphs are among the most famous classes of graphs.  Unsurprisingly, then, the problem for $k=1$ on interval graphs was studied long before the general case.  In particular, Chang~\cite{ChangSJC1998} shows that a total $1$-dominating set of minimum cardinality can be obtained in $O(n)$ time when an interval model of $G$ is given.  The huge gap in the complexities of the algorithms by Chang, on the one hand, and Chiarelli et al.\ and Kang et al., on the other hand, suggests that there is still room for improvements when $k > 1$.  One reason to explain this gap is the fact that the problems attacked by Kang et al.\ and Chiarelli et al.\ are too general.  In this note we consider the problem from the opposite perspective, by studying the simplest case that is still unsolved.  Specifically, we consider the total $2$-domination problem on proper interval graphs, for which we obtain a quadratic ($O(m)$) time algorithm.  

Our algorithm, as well as the one by Chiarelli et al.~\cite{ChiarelliHartingerLeoniPujatoMilanic2018} and others, models the total $2$-dominating problem as a shortest path problem on a weighted acyclic digraph $D$.  The major difference is that, in the model by Chiarelli et al., each vertex of $D$ represents a connected set of $G$ with diameter at most $5$.  In turn, in our model each vertex represents different connected sets of varying diameters, that correspond to the weight of each outgoing edge and can be as high as $\Omega(n)$. 

In \Cref{sec:preliminaries} we introduce the terminology required throughout the paper.  Then, in \Cref{sec:nm algorithm} we show how the total $2$-domination problem is modeled as a shortest path problem on an acyclic digraph $D$ of size $O(nm)$.  We improve this model in \Cref{sec:m algorithm}, where we observe that $D$ can be compressed to an acyclic digraph $R$, of size $O(m)$, that can be computed in $O(m)$ time.  Finally, in \Cref{sec:remarks} we discuss some ideas to try to generalize our algorithm to the case $k > 2$.

\section{Preliminaries}
\label{sec:preliminaries}

In this article we work with simple graphs and digraphs.  For a (di)graph $G$, let $V(G)$ and $E(G)$ denote the sets of vertices and (directed) edges of $G$, respectively.  As usual, we write $n = |V(G)|$ and $m = |E(G)|$ when $G$ is clear and, for simplicity, we use $vw$ to denote both the set $\{v,w\}$ and the ordered pair $(v,w)$.  Two vertices $v$ and $w$ are \emph{adjacent} when either $vw \in E(G)$ or $wv \in E(G)$.  The \emph{neighborhood} $N_G(v)$ of $v$ is the set of all its adjacent vertices, while its \emph{degree} is $d_G(v) = |N_G(v)|$.  When $G$ is a digraph, we say that $vw \in E(G)$ \emph{goes from} $v$ \emph{to} $w$, while $w$ is an \emph{out-neighbor} of $v$.  The \emph{out-degree} of $v$ is the number $d_G^+(v)$ of out-neighbors of $v$. For the sake of notation, we omit the subscript $G$ from $N$ and $d$ when no confusions are possible.

A \emph{path} in a (di)graph $G$ is a sequence of vertices $P = v_1, \ldots, v_{k+1}$ such that $v_{i}v_{i+1} \in E(G)$, for $1 \leq i \leq k$.  A \emph{cycle} is a sequence $v_1, \ldots, v_{k+1}$ such that $v_1 = v_{k+1}$ and $v_1, \ldots, v_k$ is a path.  If $G$ has no cycles, then $G$ is \emph{acyclic}.  We say that $G$ is \emph{weighted} to mean that each $e \in E(G)$ has a weight $\Weight(e) \geq 0$.  The \emph{weight} of a path $P = v_1, \ldots, v_{k+1}$ is, then, $\Weight(P) = \sum_{i=1}^k \Weight(v_iv_{i+1})$.  When $G$ is a digraph, its \emph{underlying graph} $H$ has $V(G)$ as its vertex set, whereas $vw \in E(H)$ if and only if $v$ and $w$ are adjacent in $G$, for $v,w \in V(G)$.  A graph $G$ is \emph{connected} when there is a path between every pair of vertices, while a digraph is \emph{connected} when its underlying graph is connected.  An \emph{oriented graph} is a digraph $G$ such that either $vw \not\in E(G)$ or $wv \not\in E(G)$, for $v,w \in V(G)$.  If $G$ is the underlying graph of an oriented graph $\vec{G}$, then $\vec{G}$ is an \emph{orientation} of $G$.

Consider a (di)graph $G$.  For $W \subseteq V(G)$, let $G[W]$ denote the sub(di)graph of $G$ induced by $W$.  We say that $W$ is \emph{connected} when $G[W]$ is connected.  A \emph{block} of $W$ is a connected subset of $W$ that is maximal by inclusion.  We say that $v \in V(G)$ is \emph{$2$-dominated} by $W$ when $|N(v) \cap W| \geq 2$.  If every vertex of $G$ is $2$-dominated by $W$, then $W$ is a \emph{$2$-dom}.  Moreover, if $|W|$ is minimum among the $2$-doms of $G$, then $W$ is a \emph{minimum $2$-dom}.  Note that $W$ is a $2$-dom of a graph $G$ if and only it is a $2$-dom of $\vec{G}$, for any orientation $\vec{G}$ of $G$.  Thus, we may replace $G$ by $\vec{G}$ when computing a minimum $2$-dom.  From now on we safely assume that $d(v) \geq 2$ ($v \in V(G)$); otherwise $G$ has no $2$-doms.  

A \emph{straight graph} (\Cref{fig:pig}) is an oriented graph $G$ that admits a linear ordering $v_1 <_G \ldots <_G v_n$ of its vertices and a mapping $\Fr^G \colon V(G) \to V(G)$ such that:
\begin{itemize}
 \item $v_i \leq_G \Fr^{G}(v_i)$ for $1 \leq i \leq n$ and $\Fr^{G}(v_i) \leq_G \Fr^{G}(v_{i+1})$ for $1 \leq i < n$, and
 \item $v_iv_j \in E(G)$ if and only if $v_i <_G v_j \leq_G \Fr^{G}(v_i)$.
\end{itemize}
As before, $G$ is usually omitted from $<$ and $\Fr$.  A graph $G$ is a \emph{proper interval} (PIG) \emph{graph} is some of its orientations is a straight graph.  It is well known that $G$ is a PIG graph if and only if its vertices can be mapped into a family of inclusion-free intervals of the real line in such a way that two vertices of $G$ are adjacent when their corresponding intervals have a nonempty intersection (e.g., \cite{DengHellHuangSJC1996}; see \Cref{fig:pig}).  Yet, in this article we prefer the combinatorial view provided by straight graphs.

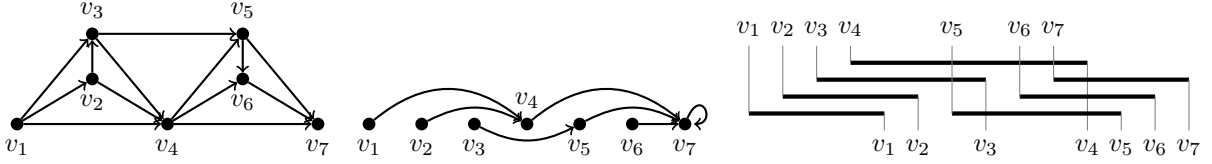
\begin{figure}
\centering
\begin{tikzpicture}[yscale=.6]
 \foreach \i/\x/\y/\l in {1/0/0/below,2/1/1/below,3/1/2/above,4/2/0/below,5/3/2/above,6/3/1/below,7/4/0/below}{%
    \node [Vertex] (\i) at (\x,\y) [label=\l:$v_\i$]{};%
 }%
 \foreach \i/\t in {1/4,2/4,3/5,4/7,5/7,6/7}{%
   \pgfmathtruncatemacro\s{\i+1}%
   \foreach \j in {\s,...,\t}{%
      \draw [thick,->] (\i) to (\j);%
   }%
 }%
\end{tikzpicture}
\begin{tikzpicture}[xscale=0.7]
 \foreach \i in {1,2,3,5,6,7}{%
    \node [Vertex] (\i) at (\i,0) [label=below:$v_\i$]{};%
 }%
 \node [Vertex] (4) at (4,0) [label=above:$v_4$]{};%
 \foreach \i/\t/\b in {1/4/30,2/4/20,3/5/-20,4/7/30,5/7/20,6/7/0}{%
   \draw [thick,->,bend left=\b] (\i) to (\t);%
 }%
 \path[thick,->,in=0,out=60] (7) edge[loop] ();
\end{tikzpicture}
\begin{tikzpicture}[ultra thick,scale=.45,yscale=.5]
    \def\Intervals{%
          1/ 5/1/1,%
          2/ 6/2/2,%
          3/ 8/3/3,%
          4/11/4/4,%
          7/12/1/5,%
          9/13/2/6,%
         10/14/3/7%
    }
    
    \foreach \s/\t/\l/\v in \Intervals {
        \draw (\s,\l) to (\t, \l);
        \draw [thin,gray] (\s,\l) to (\s, 5) node [above,black] {$v_\v$};
        \draw [thin,gray] (\t,\l) to (\t, 0) node [below,black] {$v_\v$};
    }    
\end{tikzpicture}

\caption{A straight orientation of a PIG graph $G$ (left) defined by an ordering $<$ and mapping $\Fr$ (center), and a corresponding family of inclusion-free intervals (right) that represent $G$.}\label{fig:pig}
\end{figure}

For the sake of notation, we sometimes assume that a straight graph $G$ with vertices $v_1 < \ldots < v_n$ has two \emph{artificial} vertices $v_0$ and $v_{n+1}$.  Thus, for $0 \leq i \leq j \leq n+1$, we can conveniently define the subsequence $G(v_i,v_j) = v_{i+1}, \ldots, v_{j-1}$.  For $1 \leq i \leq n$, let $\LL^G(v_i) = v_{i-1}$, $\RR^G(v_i) = v_{i+1}$, and $\Ur^G(v_i) = \RR^G(\Fr(v_i))$, where the superscript $G$ is omitted as usual.  In colloquial terms, $\LL(v)$ and $\RR(v)$ are the vertices that precede and follow $v$, respectively, and $\Ur(v)$ is the first vertex not adjacent to $v$ in $G(v,v_{n+1})$.  Also, define $\Fr^{0}(v) = v$ and $\Fr^{i+1}(v) = \Fr(\Fr^{i}(v))$ for $i > 0$.  We extend $<$ from $V(G)$ to the family of subsets of $V(G)$ in such a way that, for $W_1, W_2 \subseteq(G)$, $W_1 < W_2$ if and only if $w_1 < w_2$ for every $w_1 \in W_1$ and $w_2 \in W_2$.  It is not hard to see that the blocks of $W \subseteq V(G)$ are pairwise comparable by $<$.  Therefore, we sometimes state that $B_1 < \ldots < B_k$ are the blocks of $W$.

\section{Computing a minimum \texorpdfstring{$2$}{2}-dom in \texorpdfstring{$O(nm)$}{O(nm)} time}
\label{sec:nm algorithm}

In this section we describe an algorithm to find a minimum $2$-dom in $O(nm)$ time when a straight graph $G$ is given.  Let $v_1 < \ldots < v_n$ be the vertices of $G$, and define $S = \{v_1, v_2, v_3\}$ and $T = \{v_{n-2}, v_{n-1}, v_n\}$.  To simplify the description of the algorithm, we assume that $S$ and $T$ are blocks of $V(G)$.  Consequently, if $W$ is a $2$-dom of $G$ with blocks $B_0 < \ldots < B_{k+1}$, then $B_0 = S$ and $B_{k+1} = T$.  Note that this assumption yields no loss of generality because $B_1 \cup \ldots \cup B_k$ is a $2$-dom of $G \setminus (B_0 \cup B_{k+1})$ if and only if $B_0 \ldots \cup B_{k+1}$ is a $2$-dom of $G$.  Thus, we can always transform the input graph by inserting $S \cup T$.  The sets $S$ and $T$ are called the \emph{source} and \emph{pre-sink} blocks of $G$, respectively, whereas $v_1v_2$ and $v_{n-2}v_{n-1}$ are the \emph{source} and \emph{pre-sink} edges of $G$.

In a nutshell, the algorithm finds the blocks $B_1 < \ldots < B_k$ of the minimum $2$-dom $W$ one at a time, from $B_1$ to $B_k$.  By the discussion above, $B_1$ is simply the source block of $G$.  Once $B_i$ is determined ($1 \leq i < k$), the next block $B_{i+1}$ is obtained by choosing $j$ vertices (for some $j \geq 3$) in a way that $B_{i+1}$ reaches as far as possible.  To build $B_{i+1}$, its first two vertices $v$ and $w$ are taken as the further reaching vertices that still cover the gap from $B_i$ to $B_{i+1}$.  Then, each of the remaining $j-2$ vertices are defined in terms of $v$ and $w$.  Under the terminology defined below, $B_1, \ldots, B_k$ and $W$ are ``expansive'', while $B_{i+1}$ ``extends'' $B_i$.

Formally, a connected set $B \subseteq V(G)$ with vertices $w_1 < \ldots < w_j$ ($j \geq 3$) is \emph{expansive} (\Cref{fig:expansive}) when:
\begin{discription}[label={(exp$_\arabic*$)}]
 \item $w_3 = \Fr(w_1)$, $w_{i} = \Fr^{i-2}(w_1)$ for $4 \leq i \leq j-2$, and $w_j = \Fr(w_{j-2})$, and \label[eqitem]{def:expansive fr}
 \item if $j \geq 5$, then $w_{j-1} = \LL(w_j)$. \label[eqitem]{def:expansive > 4}
\end{discription}
By \cref{def:expansive fr,def:expansive > 4}, $B$ is fully determined by $w_1$, $w_2$, and $|B|$; for this reason, we say that $B$ is \emph{represented} by $w_1w_2$.  Clearly, $w_1w_2$ represents at most one expansive connected set of size $j$, $j \geq 0$.  Let $u = \Fr(\Ur(w_{j-1}))$ and $z = \Fr(\Ur(w_j))$ (\Cref{fig:expansive}).  An expansive connected set $B'$ represented by $vw$ \emph{extends} $B$ when:
\begin{equation}\label{eq:extends}
vw = 
\begin{cases}
    \LL(\LL(u))\LL(u) & \text{if $u = z$ and $d^+(\LL(u)) = 1$} \\
    \LL(u)u           & \text{if ($u = z$ and $d^+(\LL(u)) > 1$) or ($u \neq z$ and $d^+(u) = 1$)} \\
    u\LL(z)           & \text{if $u \neq z$, $\Fr(u) = z$ and $d^+(u) > 1$} \\
    uz                & \text{otherwise} 
\end{cases}
\end{equation}
We refer to $vw$ as being the \emph{$j$-extension} of $w_1w_2$.  Note that the $j$-extension of $w_1w_2$, if existing, is unique, because $B$ is the unique expansive connected set with $j$ vertices that is represented by $w_1w_2$.  A set $W \subseteq V(G)$ with blocks $B_1 < \ldots < B_k$ is \emph{expansive} when:
\begin{discription}[label={(exp$_\arabic*$)},resume]
 \item $B_i$ is expansive for every $1 \leq i \leq k$, and \label[eqitem]{def:block expansive}
 \item $B_{i+1}$ extends $B_i$ for every $1 \leq i < k$.  \label[eqitem]{def:expansive extension}
\end{discription}
Moreover, if $B_1$ and $B_k$ are the source and pre-sink blocks of $G$, then $W$ is \emph{fully expansive}.

\begin{figure}
\centering
\begin{tabular}{cccc}
    \begin{tikzpicture}[xscale=.7]
        \foreach \i in {1,2,3}{%
            \node [Vertex] (\i) at (\i,0) [label=below:$w_\i$]{};%
        }%
        \draw [thick,->,bend left] (1) to node [label={[above,label distance=-5pt]:$\Fr$}]{} (3);%
    \end{tikzpicture} &
    \begin{tikzpicture}[xscale=.7]
        \foreach \i in {1,...,4}{%
            \node [Vertex] (\i) at (\i,0) [label=below:$w_\i$]{};%
        }%
        \draw [thick,->,bend left] (1) to node [label={[above,label distance=-5pt]:$\Fr$}]{} (3);%
        \draw [thick,->,bend left] (2) to node [label={[above,label distance=-5pt]:$\Fr$}]{} (4);%
    \end{tikzpicture} &
    \begin{tikzpicture}[xscale=.7]
        \foreach \i/\l in {1/1,2/2,3/3,4/j-2,6/j}{%
            \node [Vertex] (\i) at (\i,0) [label=below:$w_{\l}$]{};%
        }%
        \node [Vertex] (5) at (5,0) {};%
        \draw [thick,->,densely dotted] (3) to node [label={[above,label distance=-5pt]:$\Fr$}]{} (4);%
        \draw [thick,->] (5) to node [label={[below,label distance=-3pt]:$\RR$}]{} (6);%
        \draw [thick,->,bend left] (1) to node [label={[above,label distance=-5pt]:$\Fr$}]{} (3);%
        \draw [thick,->,bend left] (4) to node [label={[above,label distance=-5pt]:$\Fr$}]{} (6);%
    \end{tikzpicture} &
    \begin{tikzpicture}[xscale=.7]
        \foreach \i/\l in {1/w_{j-1},2/w_{j},3/,4/,5/,6/,7/u,8/z}{%
            \node [Vertex] (\i) at (\i,0) [label=below:$\l$]{};%
        }%
        \draw [thick,->,bend left] (1) to node [label={[above,label distance=-5pt]:$\Fr$}]{} (3);%
        \draw [thick,->] (3) to node [label={[below,label distance=-3pt]:$\RR$}]{} (4);%
        \draw [thick,->,bend left] (2) to node [label={[above,label distance=-5pt]:$\Fr$}]{} (5);%
        \draw [thick,->] (5) to node [label={[below,label distance=-3pt]:$\RR$}]{} (6);%
        \draw [thick,->,bend left] (4) to node [label={[above,label distance=-5pt]:$\Fr$}]{} (7);%
        \draw [thick,->,bend left] (6) to node [label={[above,label distance=-5pt]:$\Fr$}]{} (8);%
    \end{tikzpicture}\\
    (a) & (b) & (c) & (d)
\end{tabular}

\caption{The expansive connected set represented by $w_1w_2$ for sizes $3$, $4$, and $j$ are depicted in (a), (b), and (c), respectively, whereas (d) describes the vertices $u$ and $z$ used to determine the $j$-extension $vw$ of $w_1w_2$ for the case in which $u \neq z$.}\label{fig:expansive}
\end{figure}
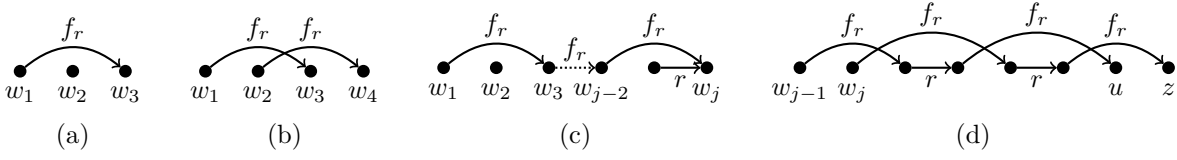

Consider the weighted digraph $D$ with vertex set $E(G)$ that has an edge from $e$ to $g$ of weight $\Weight(eg) = j$ when $g$ is the $j$-extension of $e$.\footnote{As defined, $D$ can contain multiple edges between the same pair of vertices (\Cref{fig:D}).  Moreover, some results hold only if $D$ contains these repeated edges. Yet, for simplicity, we restricted our terminology to simple digraphs.  This is not an issue, though, as all the results can be easily adapted to the case in which $D$ is simple, by ignoring the heavier repeated edges.  Is for this reason that we ignore the fact that $D$ is a multidigraph.}  Let $e$ be the pre-sink of $G$ and $t \not\in E(G)$.  Define $D(G)$ as the digraph that is obtained from $D$ after the edge $et$ with $\Weight(et) = 3$ is inserted (\Cref{fig:D}).  The vertex $t$ is the \emph{sink} of $D(G)$, while its \emph{source} is the source edge of $G$.  By definition, $B < B'$ when $B'$ extends $B$, thus $D(G)$ is acyclic.  Moreover, any path $P = e_1, \ldots, e_{k+1}$ of $D(G)$ \emph{encodes} an expansive set $W$ with blocks $B_1 < \ldots < B_k$ such that $B_i$ is represented by $e_i$ and $|B_i| = \Weight(e_ie_{i+1})$, for $1 \leq i \leq k$.  By definition, $\Weight(P) = \sum_{i=1}^{k} |B_i| = |W|$.  We record the previous discussion for later.

\begin{figure}
 \centering
 \begin{tikzpicture}[xscale=0.7,yscale=2]
 \foreach \i in {3,...,12}{%
    \node [Vertex] (\i) at (\i,0) [label=below:$\i$]{};%
    \node (anchor) at (3,-1){};
    \node (anchor2) at (3,1){};
 }%
 \foreach \i/\t/\b in {3/5/30,4/6/30,5/7/30,6/9/30,7/9/20,8/10/30,9/11/30,10/12/30,11/12/0}{%
   \draw [thick,->,bend left=\b] (\i) to (\t);%
 }%
\end{tikzpicture}
\begin{tikzpicture}[xscale=2,yscale=2]
  \foreach \i/\j/\x/\y\l in {0/1/0/1/01,3/4/0/-1/34,4/5/1/1/45,5/6/1.25/-0.6/56,6/7/4/1/67,6/8/3.5/-.2/68,7/8/2.5/-.8/78,8/9/1.25/.6/89,9/10/4/-1/9|10,10/11/0/0/10|11,13/14/2.5/0/13|14,t//3.5/0.2/t} {
    \node (\i\j) at (\x,\y) {$\l$};%
  }

  \draw [ultra thick,->] (01) to node [label={[above,label distance=-4pt]:3}] {} (45);
  \draw [thick,->] (34) to node [label={[below,label distance=-4pt]:3}] {} (910);
  \draw [thick,->] (34) to node [label={[left,label distance=-4pt]:4,5}] {} (1011);
  \draw [thick,->,bend right=25] (34) to node [pos=.6,label={[right,label distance=-7pt]:7}] {} (1314);
  \draw [ultra thick,->] (45) to node [label={[left,label distance=-4pt]:3,4}] {} (1011);
  \draw [ultra thick,->,bend left] (45) to node [label={[above,label distance=-4pt]:6}] {} (1314);
  \draw [thick,->] (56) to node [pos=.3,label={[above,label distance=-5pt]:3}] {} (1011);
  \draw [thick,->] (56) to node [pos=.3,label={[above,label distance=-5pt]:6}] {} (1314);
  \draw [thick,->,bend right=20] (67) to node [label={[above,label distance=-4pt]:5}] {} (1314);
  \draw [thick,->] (68) to node [pos=.3,label={[below,label distance=-2pt]:5}] {} (1314);
  \draw [thick,->] (78) to node [pos=.4,label={[right,label distance=-6pt]:5}] {} (1314);
  \draw [thick,->] (89) to node [label={[above,label distance=-4pt]:4,5}] {} (1314);
  \draw [thick,->] (910) to node [label={[below,label distance=-5pt]:3,4}] {} (1314);
  \draw [ultra thick,->] (1011) to node [label={[above,label distance=-4pt]:3}] {} (1314);
  \draw [ultra thick,->] (1314) to node [label={[above,label distance=-4pt]:3}] {} (t);
\end{tikzpicture}
\caption{Left: a straight graph $G$.  Right: $D(G \cup S \cup T)$ for $S = \{0,1,2\}$ and $T = \{13,14,15\}$, where bold edges belong to paths of minimum weight from the source $01$ to the sink $t$.  The three fully expansive sets encoded by $D(G)$ are $S\cup T\cup\{4,5,6,10,11,12\}$, $S \cup T \cup \{4,5,6,7,10,11,12\}$, and $S \cup T \cup \{4,5,6,9,10,11\}$.}\label{fig:D}
\end{figure}
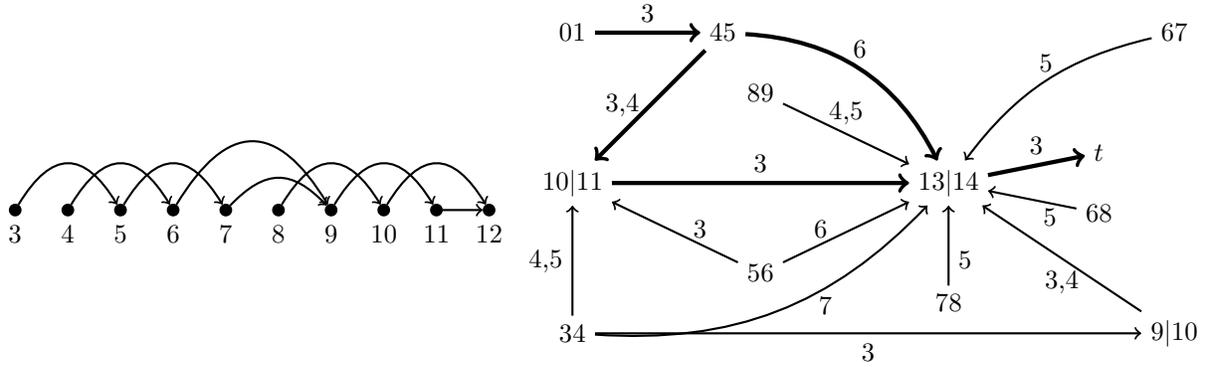

\begin{theorem}\label{thm:D(G)}
 If $G$ is a straight graph, then $D(G)$ is an acyclic digraph that has $O(m)$ vertices and $O(nm)$ edges.  Furthermore, $W \subseteq V(G)$ is (resp.\ fully) expansive if and only if\/ $W$ is encoded by a path of $D(G)$ (resp.\ from the source to the sink) whose weight is $|W|$.
\end{theorem}

The key feature about fully expansive sets is that each of them is a $2$-dom, while at least one of them is a minimum $2$-dom.  Of course, this claim holds only under our assumption that $G$ has at least one $2$-dom.

\begin{theorem}\label{thm:expansive is 2-dom}
 If\/ $W \subseteq V(G)$ is a fully expansive set of a straight graph $G$, then $W$ is a $2$-dom.
\end{theorem}

\begin{proof}
 Suppose $W = w_1 < \ldots < w_k$ and let $w_0 < w_1$ and $w_{k+1} > w_k$ be the artificial vertices of $G$.  Then, every vertex $v \in V(G)$ belongs to $G(w_i, w_j)$ for some $0 \leq i < j \leq k+1$.  Take $i$ and $j$ so that $j-i$ is minimum, and consider the following cases for $v$.
 \begin{discription}[label={\textbf{Case \arabic*:}}]
  \item $i = 0$.  This case is impossible, as $w_1w_2$ is the source edge of $G$ and, consequently, $G(w_0, w_1) = \emptyset$.
  \item $j = k+1$.  In this case, $v$ is adjacent to $w_{k-1}$ and $w_k$ because $w_{k-1}w_k$ is the pre-sink edge of $G$.
  \item $w_i$ and $w_j$ belong to the same block of $W$.  Then, either $v = w_{i+1}$ and $j = i+2$ or $j = i+1$.  Whichever the case, $v$ is adjacent to $w_i$ and $w_j$.
  \item $v = w_{i+1}$ is the first of its block.  Then, $v$ is adjacent to $w_{i+2}$ and $w_{i+3} = \Fr(w_{i+1})$.
  \item $j = i+1$ and $w_{i}$ and $w_{i+1}$ belong to different blocks.  By \cref{eq:extends}, $w_{i+1} \leq \Fr(\Ur(w_{i-1}))$ and $w_{i+2} \leq \Fr(\Ur(w_i))$.  Thus, either $v < \Ur(w_{i-1})$ is adjacent to $w_{i-1}$ and $w_i$ or $v \in G(\Fr(w_{i-1}), \Ur(w_i))$ is adjacent to $w_i$ and $w_{i+1}$ or $v \in G(\Fr(w_i), w_{i+1})$ is adjacent to $w_{i+1}$ and $w_{i+2}$. 
 \end{discription}
 As $v$ is $2$-dominated by $W$ in every case, it follows that $W$ is a $2$-dom.
\end{proof}

\begin{theorem}\label{thm:minimum 2-dom is expansive}
 If a straight graph $G$ has a $2$-dom, then it has a minimum $2$-dom that is fully expansive.
\end{theorem}

\begin{proof}
 Suppose $v_1 < \ldots < v_n$ are the vertices of $G$, and let $\Pos(v_i) = i$, $1 \leq i \leq n$.  For $W \subseteq V$, let $\Pos(W) = \sum_{w\in W}{\Pos(w)}$.  We shall prove that every minimum $2$-dom $W$ with maximum $\Pos$ is fully expansive.

 Consider any block $B$ of $W$ with vertices $w_1 < \ldots < w_j$ and let $u_3 = \Fr(w_1)$, $u_i = \Fr(w_{j-1})$ for $4 \leq i \leq j-2$, $u_j = \Fr(w_{j-2})$, and $u_{j-1} = \LL(w_{j})$ if $j \geq 5$.  Following the same pattern as in \Cref{thm:expansive is 2-dom}, it is not hard to see that $W_i = (W \setminus \{w_i\}) \cup \{u_i\}$, $3 \leq i \leq j$, is a $2$-dom of $G$.   Moreover, since $|W| \leq |W_i|$ it follows that $u_i \not\in W \setminus \{w_i\}$, hence $\Pos(W_i) = \Pos(W) + \Pos(u_i) - \Pos(w_i) \geq \Pos(W)$.  Consequently, $u_i = w_i$ by the maximality of $\Pos(W)$.  That is, $B$ satisfies \Cref{def:expansive fr,def:expansive > 4} and, thus, $W$ satisfies \Cref{def:block expansive}.
 
 Suppose now that $w_j$ is not the maximum of $W$.  Then, some expansive block $B'$ represented by an edge $vw$ appears immediately after $B$ in $W$.  Let $u = \Fr(\Ur(w_{j-1}))$ and $z = \Fr(\Ur(w_j))$, and note that $\Ur(w_j) \leq v$ because $B \cup B'$ is not connected.  Since $\Ur(w_{j-1})$ has at most one neighbor in $B$ and $\Ur(w_j)$ has no neighbors in $B$, it follows that (a) $v \leq u$ and (b) $w \leq z$.  Moreover, (c) $\Fr(v) \in B'$ by \Cref{def:expansive fr}.  Suppose that (d) $B'$ does not extend $B$, and consider the following cases.
 \begin{discription}[label={\textbf{Case \arabic*:}}]
  \item $u = z$ and $d^+(\LL(u)) = 1$.  Since $d^+(v) \geq 2$, then $v \neq \LL(u)$.  Then, by (a)~and~(b), it follows that $v \leq \LL(\LL(u))$.  Moreover, as $v$ has at least two neighbors in $W$, it follows that $w \leq \LL(u)$.  Note that $W \setminus \{v, w\} \cup \{\LL(\LL(u)), \LL(u)\}$ is a $2$-dom by (c) that, by (d) and \cref{eq:extends}, has $\Pos > \Pos(W)$.
  \item $u = z$ and $d^+(\LL(u)) > 1$.  In this case, $v \leq \LL(u)$ by (a)~and~(b), while (c) implies that $W \setminus \{v, w, \Fr(v)\} \cup \{\LL(u), u, \Fr(\LL(u))\}$ is a $2$-dom that, by (d) and \cref{eq:extends}, has $\Pos > \Pos(W)$.
  \item $u \neq z$ and $d^+(u) = 1$.  Since $d^+(v) \geq 2$, then $v \neq u$.  Then $v \leq \LL(u)$ by (a), while $w \leq u$ because $v$ has at least two neighbors in $W$.  Then, $(W \setminus \{v,w\}) \cup \{\LL(u), u\}$ is a $2$-dom by (c) that has $\Pos > \Pos(W)$ by (d) and \Cref{eq:extends}.
  \item $\Fr(u) = z$ and $d^+(u) > 1$.  Since $v$ has at least two neighbors in $W$, (b) implies $w \leq \LL(z)$.  Then, $(W \setminus \{v,w\}) \cup \{u,\LL(z)\}$ is a $2$-dom by (c) that has $\Pos > \Pos(W)$ by (d) and \Cref{eq:extends}. 
  \item $u < z < \Fr(u)$.  In this final case, $W \setminus \{v, w, \Fr(v)\} \cup \{u, z, \Fr(u)\}$ is a $2$-dom by (a)--(c) that, by (d) and \Cref{eq:extends}, has $\Pos > \Pos(W)$.
 \end{discription}
 As all the cases are impossible, $B'$ extends $B$.  Hence, \ref{def:expansive extension} holds as well.
\end{proof}

\Cref{thm:D(G),thm:minimum 2-dom is expansive,thm:expansive is 2-dom} imply that a minimum $2$-dom can be obtained by computing a path of minimum weight from the source of $D(G)$ to its sink.  By~\Cref{thm:D(G)}, this algorithm requires $O(nm)$ time once $D(G)$ is given.  We remark that $D(G)$ can be generated in $O(nm)$ time, although the details are omitted as they are similar to those discussed in the next section for $R(G)$.

\section{Computing a minimum \texorpdfstring{$2$}{2}-dom in \texorpdfstring{$O(m)$}{O(m)} time}
\label{sec:m algorithm}

The idea to accelerate the algorithm is to compress $D(G)$ in a reduced graph $R(G)$ that uses two vertices per edge of $G$.  For the sake of notation, let $\ER(G) = \ES(G) \cup \EB(G)$ for $\ES(G) = \{\SmE{e} \mid e \in E(G)\}$ and $\EB(G) = \{\BigE{e} \mid e \in E(G)\}$.  Define the \emph{width} of $vw \in E(G)$ as the minimum $\Width \geq 1$ such that $d^+(\Fr^\Width(v)) \geq 2$; when no such $\Width$ exists, the width of $vw$ is $\Width = \infty$.  

Let $s$ and $e$ be the source and pre-sink of $G$, respectively, and $t \not\in E(G)$.  As $D(G)$, the digraph $R(G)$ is obtained by inserting an edge $\SmE{e}t$ of weight $\Weight(\SmE{e}t) = 3$ in a digraph $R$ that, this time, has vertex set $\ER(G)$.  The vertices $\SmE{s}$ and $t$ are the \emph{source} and \emph{sink} of $R(G)$.  For each $e \in E(G)$ and $j \in \{3,4\}$, $R$ has \emph{regular} edges $\SmE{e}\SmE{g}$ and $\SmE{e}\BigE{g}$ of weight $j$ for each $j$-extension $g$ of $e$.  Similarly, if $e$ has width $\Width < \infty$, then $R$ has \emph{regular} edges $\BigE{e}\SmE{g}$ and $\BigE{e}\BigE{g}$ of weight $\Width+4$ when $e$ has a $(\Width+4)$-extension $g$.  This time, however, the $j$-extensions of $e$ for $j > \Width+4$ are compacted in a single edge.  Specifically, if $e = vw$, $\Fr(v) \neq w$, and $z = \Fr^{\Width}(v)$, then $R$ has a \emph{compact} edge $\BigE{e}\BigE{g}$ of weight $\Width$ for $g = z\RR(z)$.  \Cref{fig:R} depicts $R(G)$ for the straight graph $G$ in \Cref{fig:D}.

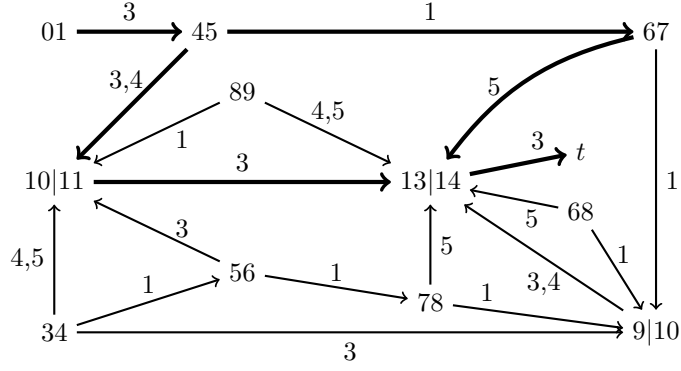
\begin{figure}
 \centering
\begin{tikzpicture}[xscale=2,yscale=2]
  \foreach \i/\j/\x/\y\l in {0/1/0/1/01,3/4/0/-1/34,4/5/1/1/45,5/6/1.25/-0.6/56,6/7/4/1/67,6/8/3.5/-.2/68,7/8/2.5/-.8/78,8/9/1.25/.6/89,9/10/4/-1/9|10,10/11/0/0/10|11,13/14/2.5/0/13|14,t//3.5/0.2/t} {
    \node (\i\j) at (\x,\y) {$\l$};%
  }

  \draw [ultra thick,->] (01) to node [label={[above,label distance=-4pt]:3}] {} (45);
  \draw [thick,->] (34) to node [label={[below,label distance=-4pt]:3}] {} (910);
  \draw [thick,->] (34) to node [label={[left,label distance=-4pt]:4,5}] {} (1011);
  \draw [thick,->] (34) to node [label={[above,label distance=-4pt]:1}] {} (56);
  \draw [ultra thick,->] (45) to node [pos=.3,label={[left,label distance=-4pt]:3,4}] {} (1011);
  \draw [ultra thick,->] (45) to node [label={[above,label distance=-4pt]:1}] {} (67);
  \draw [thick,->] (56) to node [pos=.3,label={[above,label distance=-5pt]:3}] {} (1011);
  \draw [thick,->] (56) to node [label={[above,label distance=-5pt]:1}] {} (78);
  \draw [ultra thick,->,bend right=20] (67) to node [pos=.7,label={[above,label distance=-4pt]:5}] {} (1314);
  \draw [thick,->] (67) to node [label={[right,label distance=-4pt]:1}] {} (910);
  \draw [thick,->] (68) to node [pos=.3,label={[below,label distance=-3pt]:5}] {} (1314);
  \draw [thick,->] (68) to node [pos=.3,label={[right,label distance=-4pt]:1}] {} (910);
  \draw [thick,->] (78) to node [label={[right,label distance=-4pt]:5}] {} (1314);
  \draw [thick,->] (78) to node [pos=.2,label={[above,label distance=-4pt]:1}] {} (910);
  \draw [thick,->] (89) to node [label={[above,label distance=-4pt]:4,5}] {} (1314);
  \draw [thick,->] (89) to node [pos=.3,label={[below,label distance=-4pt]:1}] {} (1011);
  \draw [thick,->] (910) to node [label={[below,label distance=-6pt]:3,4}] {} (1314);
  \draw [ultra thick,->] (1011) to node [label={[above,label distance=-4pt]:3}] {} (1314);
  \draw [ultra thick,->] (1314) to node [pos=.7,label={[above,label distance=-4pt]:3}] {} (t);
\end{tikzpicture}
\caption{$R(G \cup S \cup T)$ for $G$ in \Cref{fig:D}. For simplicity, $\protect\SmE{e}$ and $\BigE{e}$ are depicted as one vertex $e$ for every $e \in E(G)$.  The edges of weight $3$ and $4$ are from $\protect\SmE{e}$ and those of weight $1$ and $5$ are from $\BigE{e}$.  Compact edges correspond to those of weight $1$ and go to $\BigE{e}$.  Again, bold edges belong to paths of minimum weight from the source $\protect\SmE{01}$ to the sink $t$.  The three fully expansive sets encoded by $R(G)$ are the same as those in \Cref{fig:D}.}\label{fig:R}
\end{figure}

The main feature of $R(G)$ is that it preserves the adjacencies and distances of $D(G)$.  To make this assertion explicit, say that a path $P = \ReE{e}_1, \ldots, \ReE{e}_{h+1}$ of $R(G)$ is a \emph{$D$-path} when $\ReE{e}_h\ReE{e}_{h+1}$ is regular, while it is a \emph{$D$-edge} when it is a $D$-path and $\ReE{e}_i\ReE{e}_{i+1}$ is compact for $1 \leq i < h$.  Clearly, any $D$-path $P$ is equal to $P_1, \ldots, P_k$, where $P_i$ is a $D$-edge from $\ReE{e}_i$ to $\ReE{e}_{i+1}$, for $1 \leq i \leq k$.  By definition, $\ReE{e}_i \in \{\SmE{e}_i, \BigE{e}_i\}$ for some $e_i \in E(G)$.  Following the terminology for $D(G)$, we say that $P$ \emph{encodes} the expansive set $W$ with blocks $B_1 < \ldots < B_k$ such that $B_i$ is represented by $e_i$ and $|B_i| = \Weight(P_i)$, for $1 \leq i \leq k$.  \Cref{thm:R(G)} below is the translation of \Cref{thm:D(G)} to $R(G)$, that shows that $R(G)$ is actually a compact version of $D(G)$.

\begin{theorem}\label{thm:distance R(G)}
 Let $G$ be a straight graph, $e,g \in E(G)$, $\ReE{g} \in \{\SmE{g}, \BigE{g}\}$, and $j \in \Naturals$.  If $j < 5$, let $\ReE{e} = \SmE{e}$; otherwise, let $\ReE{e} = \BigE{e}$.  Then, $g$ is the $j$-extension of $e$ if and only if there exists a $D$-edge from $\ReE{e}$ to $\ReE{g}$ of weight $j$ in $R(G)$.
\end{theorem}

\begin{proof}
 Suppose first that $g$ is the $j$-extension of $e$, and let $\Width$ be the width of $e$.  We prove by induction on $j$ that $R(G)$ has a $D$-edge of weight $j$ from $\ReE{e}$ to $\ReE{g}$.  The base case, in which $j \in \{3,4,\Width+4\}$, is trivial as $\ReE{e}\ReE{g}$ is a regular edge of $R(G)$ with weight $j$.  For the inductive step, let $B = w_1 < \ldots < w_j$ be the expansive connected set of size $j \geq 5$ that is represented by $e = w_1w_2$.  Note that $B$ exists because otherwise $e$ would have no $j$-extension.  By \cref{def:expansive fr}, $w_{j-2} = \Fr^{j-4}(v)$, thus $\Width < j-4$ as $d^+(w_{j-2}) \geq 2$ and $j \neq \Width+4$.  By definition, $R(G)$ has a compact edge from $\ReE{e}$ to $\BigE{e}_\Width$, for $e_\Width = w_{\Width+2}\RR(w_{\Width+2})$, because $w_{\Width+2} = \Fr^{\Width}(v)$ by~\cref{def:expansive fr}.  Moreover, by \cref{def:expansive fr,def:expansive > 4}, $B' = w_{\Width+2}, \RR(w_{\Width+2}), w_{\Width+3}, \ldots, w_j$ is an expansive connected set with at least five vertices.  By definition, $g$ is the $(j-\Width)$-extension of $e_\Width$, thus, by induction, there is a $D$-edge $P$ from $\BigE{e}_\Width$ to $\ReE{g}$ in $R(G)$ with $\Weight(P) = j-\Width$.  Hence, $\ReE{e}P$ is a $D$-edge of $R(G)$ from $\ReE{e}$ to $\ReE{g}$ with $\Weight(\ReE{e}P) = j$.  
 
 For the converse, suppose that $R(G)$ has a $D$-edge $P = \ReE{e}_0, \ldots, \ReE{e}_h$ from $\ReE{e}_0 = \ReE{e}$ to $\ReE{e}_h = \ReE{g}$ whose weight is $j$.  Note that, by definition, $\ReE{e}_i = \BigE{e}_i$ for every $1 < i < h$.  We prove by induction on $h$ that $g$ is the $j$-extension of $e$.  The base case $h = 1$ is trivial, because $\ReE{e}\ReE{g}$ is a regular edge of $R(G)$ only if $g$ is the $j$-extension of $e$.  For the inductive step, let $\Width$ be the width of $e = w_1w_2$ and recall that $e_1 = w_{\Width+2}\RR(w_{\Width+2})$, where $w_{\Width+2} = \Fr^{\Width}(w_1)$.  By definition, $\ReE{e}_1, \ldots, \ReE{e}_h$ is a $D$-edge of $R(G)$ with weight $(j-\Width)$, which implies that $g$ is the $(j-\Width)$-extension of $e_1$ by induction.  Note that $j-\Width \geq 5$ because $\ReE{e}_1 = \BigE{e}_1$.  Thus, by~\cref{def:expansive fr,def:expansive > 4}, $e_1$ represents an expansive connected set $B' = w_{\Width+2} < \RR(w_{\Width+2}) < w_{\Width+3} < \ldots < w_j$ such that $w_{i} = \Fr^{i-\Width}(w_{\Width+2}) = \Fr^{i-2}(w_1)$ for every $\Width \leq i \leq j-2$, $w_{j} = \Fr(w_{j-2})$, and $w_{j-1} = \LL(w_j)$.  Therefore, by~\cref{def:expansive fr,def:expansive > 4}, $B = w_1,\ldots, w_j$ is an expansive connected set with $|B| = j$ when $w_i = \Fr^{i-2}(w_1)$ for $3 \leq i \leq \Width$.  Consequently, by~\cref{eq:extends}, $e$ has $g$ as its $j$-extension.
\end{proof}

\begin{theorem}\label{thm:R(G)}
 If $G$ is a straight graph, then $R(G)$ is an acyclic digraph that has $O(m)$ vertices and edges.  Furthermore, $W \subseteq V(G)$ is (resp.\ fully) expansive if and only if\/ $W$ is encoded by a $D$-path of $R(G)$ (resp.\ from the source to the sink) whose weight is $|W|$.
\end{theorem}

\begin{proof}
 By \Cref{thm:D(G)}, any expansive set $W$ is encoded by a path $P = e_1, \ldots, e_{k+1}$ of $D(G)$.  Let $h = k$ if $e_{k+1}$ is the sink of $D(G)$ and $h = k+1$ otherwise.  For $1 \leq i \leq h$, let $\ReE{e}_i = \SmE{e}_i$ if $\Weight(e_ie_{i+1}) < 5$ and $\ReE{e}_i = \BigE{e}_i$ otherwise.  By Theorem~\ref{thm:distance R(G)}, there is a $D$-edge $P_i$ from $\ReE{e}_i$ to $\ReE{e}_{i+1}$ of weight $\Weight(e_ie_{i+1})$ for every $1 \leq i < h$.  If $h = k$, then $\ReE{e}_k = \SmE{e}_k$  because the unique edge form the pre-sink of $G$ in $D(G)$ has weight $3$.  Thus, regardless of the value of $h$, the edge of $R(G)$ from $\ReE{e}_k$ to $\ReE{e}_{k+1}$ is a $D$-edge of weight $\Weight(e_{k}e_{k+1})$.  Consequently, $P_1, \ldots, P_k$ is a $D$-path of $R(G)$ that encodes $W$.  Moreover, if $e_1$ is the source edge of $G$, then $\Weight(e_1e_2) = 3$ and, therefore, $\ReE{e}_1 = \SmE{e}_1$ is the source of $R(G)$ by \Cref{thm:distance R(G)}.  Hence, by \Cref{thm:D(G)}, $P_1, \ldots, P_k$ goes from the source of $R(G)$ to its sink when $W$ is fully expansive.
 
 The converse is similar: if $P_1, \ldots, P_k$ is a $D$-path of $R(G)$ that encodes a set $W$, then $P = e_1, \ldots, e_{k+1}$ is a path of $D(G)$ by \Cref{thm:distance R(G)}, where $e_i \in E(G)$ is the edge corresponding to the first edge of $P_i$ for $1 \leq i \leq k$, and $e_{k+1}$ corresponds to last edge of $P_{k}$ that happens to be the sink of $D(G)$ when $P_k$ ends at the sink of $R(G)$.  Moreover, $\Weight(P_i) = \Weight(e_ie_{i+1})$.  Thus, $P$ encodes $W$ which, by \Cref{thm:D(G)}, implies that $W$ is an expansive set of $G$ with $\Weight(P_1, \ldots, P_k)$ vertices.  Moreover, $W$ is fully expansive when $P_1, \ldots, P_k$ goes from the source to the sink of $R(G)$.
\end{proof}

The algorithm to compute a minimum $2$-dom of a given straight graph $G$ has three main steps.  \Cref{step:R(G),step:P} compute $R(G)$ and a path $P$ of minimum weight from the source to the sink of $R(G)$, respectively.  By \Cref{thm:distance R(G),thm:expansive is 2-dom,thm:minimum 2-dom is expansive}, $P$ encodes a minimum $2$-dom $W$ of $G$; the set $W$ is found in Step~\ref{step:W}.  The algorithm runs in $O(m)$ time when implemented as described below, where we write $[n] = [1,n]\cap \Naturals$.
\begin{discription}[label={\textbf{Step \arabic*:}},ref=\arabic*]
 \item [\bf Input:] $G$ is implemented with the sequence $v_1 < \ldots < v_n$ of its vertices and a function $\FrImpl\colon [n] \to [n]$ such that $\FrImpl(i) = j$ when $\Fr(v_i) = v_j$.  Both $V(G)$ and $\FrImpl$ are implemented with vectors, thus traversing $V(G)$ requires $O(n)$ time, whereas querying $\FrImpl(i)$ costs $O(1)$ time.  Note that $\hat{d}^+(i) = d^+(v_i) = \FrImpl(i) - i$ can be answered in $O(1)$ time as well.  We assume that $G$ contains the source and pre-sink blocks, as $O(n)$ time suffices to insert them into the structure.
 \item [\bf Step 0:] before computing $R(G)$, we build the map $\WidthImpl\colon [n] \to [n+1] \times [n]$ such that: if $v_i$ has width $\Width <\infty$, then $\WidthImpl(i) = (\Width, \Fr^{\Width}(v_i))$; otherwise, $\WidthImpl(i) = (n+1,i)$.  A single backward traversal of $V(G)$ suffices to compute $\WidthImpl$ in $O(n)$ time because, by definition,
 \begin{displaymath}
  \WidthImpl(i) = 
  \begin{cases}
   (n+1,i)                 & \text{if $\hat{d}^+(\FrImpl(i)) = 0$} \\
   (1, \FrImpl(i))         & \text{if $\hat{d}^+(\FrImpl(i)) \geq 2$} \\
   (1+\min\{\Width,n\}, w) & \text{otherwise, for $(\Width, w) = \WidthImpl(\FrImpl(i))$.}
  \end{cases}
 \end{displaymath}
 \item\label[step]{step:R(G)} to compute $R(G)$, first two vertices $(i,j,0)$ and $(i,j,1)$ representing $\SmE{e}$ and $\BigE{e}$ are created, for $e = v_iv_j \in E(G)$, $i \in [n]$, and $j \in (i,\FrImpl(i)]$.  This step consumes $O(m)$ time.  Then, for $i \in [n]$ and $j \in (i, \Fr(i))$, the edges of $R(G)$ from $(i,j,1)$ are inserted.  Let $(\Width, a) = \WidthImpl(i)$ and $b = \FrImpl(a)$; suppose $\Width \leq n$ as no edge from $(i,j,1)$ has to be inserted otherwise.  First, the edge from $(i,j,1)$ to $(a,a+1,1)$, representing the compact edge $\BigE{e}\BigE{g}$ for $e = v_iv_j$ and $g = v_a\RR(v_a)$, is inserted in $O(1)$ time.   To create the regular edges, note that, by~\cref{def:expansive fr}, $v_{b-1}$ and $v_b$ are the last two vertices of the expansive connected set of size $(\Width+4)$ that is represented by $v_iv_j$.  Clearly, the indices $x,y$ such that $v_xv_y$ is the edge defined by \cref{eq:extends}, when applied to $v_{b-1}$ and $v_b$, can be obtained in $O(1)$ time with a few applications of $\FrImpl$.  By definition, $v_xv_y$ is the $(\Width+4)$-extension of $v_iv_j$ if and only if $v_bv_x \not\in E(G)$ and $v_xv_y$ represents an expansive connected set.  These facts that can be determined in $O(1)$ time by observing whether $\FrImpl(b) < x$ and $\hat{d}^+(x) \geq 2$.  If affirmative, then the edges from $(i,j,1)$ to $(x,y,\bullet)$ of weight $\Width+4$ are inserted in $O(1)$ time.  Therefore, the edges from $(i,j,1)$ are created in $O(1)$ time.  Each regular edge from $(i,j,0)$ is inserted $O(1)$ time with a similar procedure.  Therefore, this step consumes $O(m)$ total time.
 \item\label[step]{step:P}  since $R(G)$ is acyclic (\Cref{thm:distance R(G)}), $P$ can be computed in $O(m)$ time.
 \item\label[step]{step:W}  traversing $P$ once, we can split $P$ into $D$-edges $P_1, \ldots, P_k$ while $\Weight(P_h)$ is computed for $h \in [k]$.  Let $(i,j,\bullet)$ be the first vertex of $P_h$.  By~\Cref{def:expansive > 4,def:expansive fr}, $O(\Weight(P_h))$ applications of $\FrImpl$ from $i$ are enough to find all the vertices of the expansive connected set $B_h$ of size $\Weight(P_h)$ that is represented by $v_iv_j$.  Then, the output $W = B_1 \cup \ldots \cup B_k$ can be computed in $O(\Weight(P)) = O(n)$ time.
\end{discription}
Note that $G$ can be encoded in $O(n)$ space, thus the algorithm is quadratic in the worst case.  We remark that many algorithms exist to compute a straight orientation $\vec{G}$ of a PIG graph $G$ in $O(m)$ time.  In particular, the algorithm in~\cite{DengHellHuangSJC1996} outputs $\vec{G}$ as required by the algorithm above.  Thus, when $G$ is a PIG graph represented with adjacency lists, a $2$-dom can be computed in linear time.

\section{Concluding remarks}
\label{sec:remarks}

In this note we developed an $O(m)$ time algorithm for the total $2$-dominating set problem on proper interval graphs, improving the previous $O(n^6)$ time algorithm by Chiarelli et al.~\cite{ChiarelliHartingerLeoniPujatoMilanic2018}.  Both of these algorithms work by finding a shortest path on a weigthed digraph $D$.  The main difference between them is that in our model the edges of $D$ represent connected sets with a large diameter.  The actual connected set represented by $e \in D$ is the one that reaches farther in the input (model of the) graph $G$.  One of the consequences defining the edges of $D$ in this way is that some connected sets that can be a part of the solution when $G$ is weighted are not considered.  Therefore, on the contrary to the algorithm by Chiarelli et al., our algorithm does not solve the problem when $G$ is weighted.

Our algorithm provides more evidence that the time required to solve problem of finding a total $k$-dominating set on a (proper) interval graph, for $k > 2$, is $o(n^{3k})$.  In our digraph $D$, each edge goes from a pair $vw$ to the another pair $uz$.  Certainly, we can extend this model to $k$-tuples; the idea would be to have an edge from a $k$-tuple $v$ to a $k$-tuple $w$ of weight $j$ when $w$ is the further tuple that can be reached with a ``block'' having $j$ vertices.  Intuitively, such a $k$-tuple $w$ should exist: if $B$ and $B'$ are two ``blocks'' of a total $k$-dominating set that begin with $v$ and neither of them is lexicographically larger than the other, then it should be possible to combine $B$ and $B'$ into a new block beginning with $v$ that is lexicographically larger than both $B$ and $B'$.  Thus, the tuple $w$ reaching further should exist.  The problem, however, is how to compute $w$ when building $D$.  The case $k = 2$ is easy because all the blocks have a peculiar structure.  We conjecture that, by following these ideas, the problem can be solved in $O(n^kk^k)$ time.


\end{document}